\RequirePackage{fix-cm}
\documentclass[smallcondensed]{svjour3}     
\smartqed  
\usepackage{graphicx}

\usepackage{amssymb,bm,xargs}
\usepackage{booktabs}
\usepackage[ruled,vlined,linesnumbered]{algorithm2e}

\usepackage{todonotes}

\usepackage{xargs}

\usepackage{xspace}

\usepackage{xstring}

\usepackage{boolexpr}

\usepackage[cmex10]{mathtools}
\usepackage{amssymb}
\usepackage{amsfonts}
\usepackage{mathrsfs}
\usepackage{latexsym}
\usepackage{textcomp}
\usepackage{pifont}

\newcommand{\argemp}[2]{\if&#1&\else#2\fi}

\newcommand{\argdef}[2]{\if&#1&#2\else#1\fi}

\newcommand{\argint}[3]{\if&#2&\else#1#2#3\fi}

\newcommand{\argext}[3]{\if&#1&#3\else#1\if&#3&\else#2#3\fi\fi}

\newcommandx{\mthfnt}[3][1=, 2=0]{{
	\IfStrEqCase{#1}
	{%
		{}%
		{#3}%
		{Name}%
		{%
			\IfStrEqCase{#2}
			{%
				{0}{\mathcal{#3}}%
				{1}{\mathscr{#3}}%
				{2}{\mathfrak{#3}}%
				{3}{\mathbb{#3}}%
			}
			[\ensuremath{\clubsuit}]%
		}%
		{Set}%
		{%
			\IfStrEqCase{#2}
			{%
				{0}{\mathrm{#3}}%
				{1}{\mathsf{#3}}%
				{2}{\mathbb{#3}}%
				{3}{\mathbf{#3}}%
			}
			[\ensuremath{\clubsuit}]%
		}%
		{Fun}%
		{%
			\IfStrEqCase{#2}
			{%
				{0}{\mathsf{#3}}%
				{1}{\mathrm{#3}}%
			}
			[\ensuremath{\clubsuit}]%
		}%
		{Rel}%
		{%
			\IfStrEqCase{#2}
			{%
				{0}{\mathit{#3}}%
				{1}{\mathtt{#3}}%
			}
			[\ensuremath{\clubsuit}]%
		}%
		{Sym}%
		{%
			\IfStrEqCase{#2}
			{%
				{0}{\mathtt{#3}}%
				{1}{\mathbf{#3}}%
			}
			[\ensuremath{\clubsuit}]%
		}%
		{Elm}%
		{\mathnormal{#3}}
	}
[\ensuremath{\clubsuit}]%
}}

\newcommand{\mthsub}[1]{\argemp{#1}{\ensuremath{_{\mathnormal{#1}}}}}

\newcommand{\mthsup}[1]{\argemp{#1}{\ensuremath{^{\mathnormal{#1}}}}}

\newcommandx{\mth}[5][1=, 2=0, 4=, 5=]{{\ensuremath{\mthfnt[#1][#2]{#3}\mthsub{#4}\mthsup{#5}}}}

\newcommandx{\mtharg}[6][1=, 2=0, 4=, 5=]{{\mth[#1][#2]{#3}[#4][#5]\ensuremath{\argint{(}{#6}{)}}}}

\newcommand{\mthempty}{\mth[][]}

\newcommand{\mthstyname}{0}
\newcommand{\mthname}[1][]{\mth[Name][\argdef{#1}{\mthstyname}]}

\newcommand{\mthstyset}{0}
\newcommand{\mthset}[1][]{\mth[Set][\argdef{#1}{\mthstyset}]}

\newcommand{\mthstyfun}{0}
\newcommand{\mthfun}[1][]{\mth[Fun][\argdef{#1}{\mthstyfun}]}

\newcommand{\mthstysym}{0}
\newcommand{\mthsym}[1][]{\mth[Sym][\argdef{#1}{\mthstysym}]}

\newcommand{\mthstyelm}{0}
\newcommand{\mthelm}[1][]{\mth[Elm][\argdef{#1}{\mthstyelm}]}

\newcommand{\tuple}[1]
{\ensuremath{\!\argint{\langle}{#1}{\rangle}}}

\newcommand{\ignore}[1]{}

\newcommand{\SetN}{\mathbb{N}}
\newcommand{\SetZ}{\mathbb{Z}}
\newcommand{\SetR}{\mathbb{R}}

\newcommand{\set}[2]{\ensuremath{\argint{\{}{\argext{#1}{\allowbreak:\allowbreak}{#2}}{\}}}}

\newcommand{\card}[1]{\mthempty{\argint{\vert}{#1}{\vert}}}

\newcommand{\Trans}{\mthname{T}}

\newcommand{\Kripke}{\mathcal{K}}

\newcommand{\LTL}{\mthfun{LTL}\xspace}
\newcommand{\GRone}{\mthfun{GR(1)}\xspace}
\def\TEMPORAL#1{\mbox{\small\boldmath$\mathbf{#1}$}}
\def\ltlnext{\TEMPORAL{X}}
\def\sometime{\TEMPORAL{F}} 
\def\always{\TEMPORAL{G}}
\def\until{\,\TEMPORAL{U}\,}
\def\Nat{\mathbb{N}}

\newcommand{\Ag}{\mthset{N}}
\newcommand{\Ac}{\mthset{Ac}}
\newcommand{\AcProf}{\vec{\Ac}}
\newcommand{\St}{\mthset{St}}
\newcommand{\AP}{\mthset{AP}}

\newcommand{\Pun}{\text{Pun}}

\renewcommand{\Game}{\mthname{G}}

\renewcommand{\Pun}{\mthset{Pun}}
\newcommand{\pun}{\mthsym{pun}}

\newcommand{\Automaton}{\mthname{A}}

\newcommand{\labFun}{\lambda}

\newcommand{\WinSet}{\mthset{Win}}
\newcommand{\LoseSet}{\mthset{Lose}}

\newcommand{\trnFun}{\mthfun{tr}}
\newcommand{\act}{\mthsym{a}}
\newcommand{\jact}{\vec{\mthsym{a}}}

\newcommand{\StrSet}{\mthset{Str}}

\newcommand{\strElm}{\sigma}

\newcommand{\strpElm}{\mthelm{\vec{\sigma}}}

\newcommand{\NE}{\mthset{NE}}

\newcommand{\winsym}{\mthset{Win}}
\newcommandx{\Win}[3][1=, 2=, 3=]
{\mthset{\winsym#3}[#1][#2]}

\newcommand{\presym}{\mthfun{Pre}}
\newcommandx{\Pre}[3][1=, 2=, 3=]
{\mthset{\presym#3}[#1][#2]}

\newcommand{\eqsym}{\mthfun{Eq}}
\newcommandx{\Eq}[3][1=, 2=, 3=]
{\mthset{\eqsym#3}[#1][#2]}

\newcommandx{\AFW}[5][1=, 2=, 3=, 4=, 5=]
{\txtargname{AFW#5{\small\argint{$[$}{#1}{$]$}}}[#2][#3]{#4}\xspace}

\newcommand{\MP}[1][]{%
	\ifthenelse{\equal{#1}{}}{{\small{\sf mp}}}{{\small{\sf mp}}(#1)}%
\xspace}

\providecommand{\strFun}[1][]{\mthfun{\sigma}}
\providecommand{\pstrFun}[1][]{\mthfun{\sSym}}
\newcommand{\enash}{\mthfun{E\mbox{-}Nash}\xspace}
\newcommand{\anash}{\mthfun{A\mbox{-}Nash}\xspace}
\newcommand{\rvnonemptiness}{\mthfun{Non\mbox{-}emptiness}\xspace}

\newcommand{\wFun}{\mthfun{w}}

\newcommand{\CTL}{\mthfun{CTL}\xspace}
\newcommand{\CTLstar}{\mthfun{CTL^{\ast}}\xspace}

\newcommand{\LTLlim}{\mthfun{LTL^{Lim}}\xspace}

\def\avg{{\sf avg}}

\def\src{{\sf src}}
\def\trg{{\sf trg}}
\def\IN{{\sf in}}
\def\OUT{{\sf out}}

\newcommand{\pay}{\mthfun{pay}}

\newcommand{\sw}{\mthfun{sw}\xspace}
\newcommand{\esw}{\mthfun{esw}\xspace}
\newcommand{\usw}{\mthfun{usw}\xspace}
\newcommand{\MaxENash}{\mthfun{MaxNE}\xspace}
\newcommand{\MinENash}{\mthfun{MinNE}\xspace}

\newcommand{\LP}{\mthfun{LP}}

\newcommand{\exptime}{\mthfun{EXPTIME}\xspace}
\newcommand{\exptimeC}{\mthfun{EXPTIME}-complete\xspace}

\def\pspace{\mthfun{PSPACE}\xspace}
\def\pspaceC{\mthfun{PSPACE}-complete\xspace}

\def\np{\mthfun{NP}\xspace}
\def\npC{\mthfun{NP}-complete\xspace}

\def\fpt{\mthfun{FPT}\xspace}

\usepackage{tikz}
\usetikzlibrary{arrows,shapes,snakes}

\usepackage{wrapfig}

\tikzstyle{every node} =
	[draw = none, fill = white, thin]
\tikzstyle{every edge} +=
	[thin]

\tikzstyle{noall} =
	[draw = none, fill = none]
\tikzstyle{nodraw} =
	[draw = none, fill = white]
\tikzstyle{nofill} =
	[draw = black, fill = none]

\tikzstyle{cnode} =
	[circle, draw = black]
\tikzstyle{snode} =
	[regular polygon, regular polygon sides = 4, draw = gray!50]
\tikzstyle{lnode} =
	[diamond, draw = gray!75]
\tikzstyle{pnode} =
	[regular polygon, regular polygon sides = 5, draw = gray]
\tikzstyle{rnode} =
[rectangle, draw = black, thin]

	{

	}

\begin{document}

\title{On the Complexity of Rational Verification
}


\author{
	Julian Gutierrez\and
	Muhammad Najib\and
	Giuseppe Perelli\and
	Michael Wooldridge
}


\institute{
	J.\ Gutierrez \at
	Faculty of Information Technology, Monash University, Australia. \\
	\email{julian.gutierrez@monash.edu} \and
	M.\ Najib \at
	Department of Computer Science, University of Kaiserslautern, Germany. \\
	\email{najib@cs.uni-kl.de} \and	G.\ Perelli \at
	Department of Computer Science, Sapienza University of Rome, Italy. \\
	\email{perelli@di.uniroma1.it} \and
	M.\ Wooldridge \at
	Department of Computer Science, University of Oxford, United Kingdom. \\
	\email{michael.wooldridge@cs.ox.ac.uk}
}

\date{Received: date / Accepted: date}

\maketitle

\begin{abstract}
	\emph{Rational verification} refers to the problem of checking which
  temporal logic properties hold of a concurrent/multiagent system,
  under the assumption that agents in the system choose strategies
  that form a game theoretic equilibrium. Rational verification can be
  understood as a counterpart to model checking for multiagent
  systems, but while classical model checking can be done in
  polynomial time for some temporal logic specification languages such
  as \CTL, and polynomial space with \LTL\ specifications, rational
  verification is much harder: the key decision problems for rational verification
  are 2\exptimeC with \LTL specifications,
  even when using explicit-state system representations. Against this
  background, our contributions in this paper are threefold. First, we
  show that the complexity of rational verification can be greatly
  reduced by restricting specifications to \GRone, a fragment of \LTL
  that can represent a broad and practically useful class of response
  properties of reactive systems. In particular, we show that for a number of
  relevant settings, rational verification can be done in polynomial
  space and even in polynomial time. Second, we provide improved
  complexity results for rational verification when considering
  players' goals given by \emph{mean-payoff} utility functions---arguably the
  most widely used approach for quantitative objectives in concurrent and
  multiagent systems. Finally, we consider the problem of computing outcomes that satisfy social
  welfare constraints. To this end, we consider both utilitarian and egalitarian
  social welfare and show that computing such outcomes is either \pspace-complete or \np-complete.
\end{abstract}


\section{Introduction}
\label{sec:introduction}

The formal verification of computer systems has been a major research
area in computer science for the past 60 years.  Verification is the
problem of checking program correctness: the key decision problem
relating to verification is that of establishing whether or not a
given system $P$ satisfies a given specification. The most successful
contemporary approach to formal verification is model checking, in
which an abstract, finite state model of the system of interest $P$ is
represented as a Kripke structure $K_P$ (a labelled transition
system), and the specification is represented as a temporal logic
formula $\varphi$, the models of which are intended to correspond to
``correct'' behaviours of the system~\cite{emerson:90a}. The
verification process then reduces to establishing whether the
specification formula $\varphi$ is satisfied in the Kripke structure
$K_P$ (notation: $K_P\models\varphi$), a process that can be
efficiently automated in many settings of interest~\cite{CGP02}. For
example, model checking Linear Temporal Logic (\LTL) specifications
can be done in polynomial space, and for specifications in
Computation Tree Logic (\CTL) it can be done in polynomial
time~\cite{clarke:2018a}.

In the context of multiagent systems, \emph{rational verification}
forms a natural counterpart of model
checking~\cite{GutierrezHW15,WooldridgeGHMPT16,GutierrezHW17-aij}.
This is the problem of checking whether a given property $\varphi$,
expressed as a temporal logic formula, is satisfied in a computation
of a system that might be generated \emph{if agents within the system
  choose strategies for selecting actions that form a game-theoretic
  equilibrium}. This game theoretic aspect of rational verification
adds a new ingredient to the verification problem, as it becomes
necessary to take into account the {\em preferences} of players with
respect to the possible runs of the system.  Typically, in rational
verification, such preferences are given by associating an \LTL\ goal
$\gamma_i$ with each player~$i$ in the game: player $i$ prefers all
those runs of the system that satisfy $\gamma_i$ over those that do not,
is indifferent between all those runs that satisfy $\gamma_i$, and is
similarly indifferent between those runs that do not satisfy $\gamma_i$.
In this setting, rational verification with respect to a specification
$\varphi$ is 2\exptimeC, regardless of whether the representation of the
system is given succinctly~\cite{GutierrezHW17-aij,GutierrezHW15} or
explicitly simply as a finite-state labelled transition
graph~\cite{GHW15}.  This high computational complexity represents a
key barrier to the wider take-up of rational verification.

Our aim in this work is to improve this state of affairs: we present a range of
settings for which we are able to give complexity results that greatly
improve on the 2\exptimeC result of the general \LTL\ case.  We first
consider games where the goals of players are represented as {\em
  \GRone formulae}. \GRone\ is an important fragment of \LTL that can
express a wide range of practically useful response properties of
concurrent and reactive systems~\cite{BJPPS12}. We then consider {\em
  mean-payoff utility functions}: one of the most studied reward and
quality measures used in games for automated formal verification.  In
each case, we study the rational verification problem for system
specifications $\varphi$ given as \GRone formulae and as \LTL
formulae, with respect to system models that are formally represented
as concurrent game structures~\cite{AlurHK02}.

Our main results, summarised in Table~\ref{tab:results}, show that in
the cases mentioned above, the 2\exptime result can be dramatically
improved, to settings where rational verification can be solved in
polynomial space, \np, or even in polynomial time, if the number of
players in the game is assumed to be fixed.

%

\begin{table*}[!ht]
	\begin{center}
		\setlength{\tabcolsep}{12pt}
		\begin{tabular}{c c c c}
			\toprule
			Players' goals & Specification & \enash & \\ 
			\midrule
			\LTL & \LTL & 2\exptimeC & \\[.25ex] 
			\GRone & \LTL & \pspaceC &(Corollary~\ref{cor:enashgroneltl}) \\[.25ex]
			\GRone & \GRone & \fpt &(Theorem~\ref{thm:enashgronegrone}) \\[.25ex]
			\MP & \LTL & \pspaceC &(Corollary~\ref{cor:enashmpltl}) \\[.25ex]
			\MP & \GRone & \npC &(Theorem~\ref{thm:enashmpgrone}) \\[.25ex]
			\bottomrule
		\end{tabular}
		\vspace{-5pt}
	\end{center}
	\caption{Summary of main complexity results.}
	\label{tab:results}
	\vspace{-10pt}
\end{table*}

In addition to characterising the complexity of the core rational verification problems for these settings, we also consider the problem of computing strategy profiles for players that \emph{maximise social welfare}.
Measures of social welfare are measures of how well society as a whole fares with some particular game outcome; thus social welfare measures are \emph{aggregate} measures of utility.
We look at two well-known measures of social welfare: \emph{utilitarian social welfare} (in which we aim to maximise the sum of individual agent utilities) and \emph{egalitarian social welfare} (in which we try to maximise the utility of the worst-off player).
We show that, for mean payoff games, computing outcomes for these measures with \LTL specifications is \pspace-complete.

\subsection*{\bf Related Work} The rational verification problem has been studied for a number of different settings, including iterated Boolean games, reactive modules games, and concurrent game structures~\cite{GutierrezHW15,GutierrezHW17-aij,GHW15,GutierrezHW17-apal}.
In each of these settings, the main rational verification problems are 2\exptimeC, and hence highly intractable. Rational verification is closely related to rational synthesis, which is also 2\exptimeC both in the Boolean case~\cite{FismanKL10} and with rational environments~\cite{KupfermanPV16}. One might mitigate the problem of intractability by considering low-level languages such as omega-regular specifications \cite{SGW21,ConduracheOT18} and turn-based setting \cite{ConduracheFGR16}.
All of the above cases only consider perfect information.
In settings with imperfect information, the problem has been shown to be undecidable both for games with succinct and explicit model representations~\cite{GutierrezPW18,FiliotGR18}.

Our work also relates to \LTL and mean-payoff (\MP) games in
general. While the former are already 2\exptimeC even for two-player
games (and in fact already 2\exptime-hard for many \LTL
fragments~\cite{AlurT04}), the latter are \npC for multi-player
games~\cite{UW11} and in $\np \cap \mathsf{co}\np$ for two-player
games~\cite{ZP96}, and in fact solvable in quasipolynomial time since
they can be reduced to two-player perfect-information parity
games~\cite{CaludeJKLS17}. Even though we provide several complexity results that improve on the complexity of the general case, our solutions are unlikely to run in polynomial time, for instance as \CTL\ model checking, since rational verification subsumes problems that are typically not known to be solvable in polynomial time, such as model checking or automated synthesis with temporal logic specifications.






\section{Preliminaries}
\label{sec:prelims}


\noindent \textbf{Linear Temporal Logic.}  \LTL extends propositional
logic with two operators, $\ltlnext$ (``next'') and $\until$
(``until''), for expressing properties of
paths~\cite{pnueli:77a,emerson:90a}.
The syntax of \LTL is defined with respect to a set $\AP$ of atomic propositions as follows:
$$ \phi ::= 
\mathop\top \mid
p \mid
\neg \phi \mid\phi \vee \phi \mid
\ltlnext \phi \mid
\phi \until \phi
$$
where $p \in \AP$.
As usual, we define $\phi_1 \wedge \phi_2 \equiv \neg (\neg \phi_1 \vee \neg \phi_2)$, $\phi_1 \to \phi_2 \equiv \neg \phi_1 \vee \phi_2$, $\sometime \phi \equiv \mathop\top \until \phi$, and $\always \phi \equiv \neg \sometime \neg \phi$.
We interpret \LTL formulae with respect to pairs $(\alpha,t)$, where $\alpha \in (2^{\AP})^{\omega}$ is an infinite sequence of sets of atomic proposition that indicates which propositional variables are true in every time point and $t \in \Nat$ is a temporal index into $\alpha$.
As usual, by $\alpha_t \in 2^{\AP}$ we denote the $t$-th element of the infinite sequence $\alpha$.
Formally, the semantics of \LTL is given by the following rules:
$$
\begin{array}{lcl}
(\alpha,t)\models\mathop\top	\\
(\alpha,t)\models p 				&\text{ iff }&	p\in \alpha_t\\
(\alpha,t)\models\neg \phi			&\text{ iff }&   \text{it is not the case that $(\alpha,t) \models \phi$}\\
(\alpha,t)\models\phi \vee \psi		&\text{ iff }&	\text{$(\alpha,t) \models \phi$  or $(\alpha,t) \models \psi$}\\
(\alpha,t)\models\ltlnext\phi			&\text{ iff }&	\text{$(\alpha,t+1) \models \phi$}\\
(\alpha,t)\models\phi\until\psi	&\text{ iff }&   \text{for some $t' \geq t: \ \big((\alpha,t') \models \psi$  and }\\
&&\quad\text{for all $t \leq t'' < t': \ (\alpha,t'') \models \phi \big)$.}\\
\end{array}
$$
If $(\alpha,0)\models\phi$, we write $\alpha\models\phi$ and say that
\emph{$\alpha$ satisfies~$\phi$}.

\vspace*{4pt} 
\noindent \textbf{General Reactivity of rank 1.}
The language of \emph{General Reactivity of rank 1}, ($\GRone$), is the fragment of \LTL containing formulae that are written in the following form~\cite{BJPPS12}:
$$
(\always \sometime \psi_1 \wedge \ldots \wedge \always \sometime \psi_m) \to (\always \sometime \phi_1 \wedge \ldots \wedge \always \sometime \phi_n)
\text{,}
$$
where subformulae $\psi_i$ and $\phi_i$ are Boolean combinations of atomic propositions.

\vspace*{4pt} 
\noindent \textbf{Mean-Payoff value.}
For an infinite sequence $\beta \in \mathbb{R}^\omega$ of real numbers, let $\MP(\beta)$ be
denote \emph{mean-payoff} value of $\beta$, that is, 
\[ \MP(\beta) = \lim \inf_{n \to \infty} \avg_n(\beta)
\]
where, for $n \in \mathbb{N}$, we define
\[ \avg_n(\beta) = \frac{1}{n}\sum_{j=0}^{n-1} \beta_j.\] 

\vspace*{4pt} \noindent \textbf{Arenas.}
An \emph{arena} is a tuple
\[A = \tuple{\Ag,  \Ac, \St, s_0, \trnFun, \labFun}\]
where $\Ag$, $\Ac$, and $\St$ are finite non-empty sets of \emph{players} (write $N = \card{\Ag}$), \emph{actions}, and \emph{states}, respectively;	$s_0 \in \St$ is the \emph{initial state}; $\trnFun : \St \times \AcProf \rightarrow \St$ is a \emph{transition function} mapping each pair consisting of a state $s \in \St$ and an \emph{action profile} $\jact \in \AcProf = \Ac^{\Ag}$, one for each player, to a successor state; and $\labFun: \St \to 2^{\AP}$ is a \emph{labelling function}, which maps every state to a subset of \emph{atomic propositions}---the atomic propositions that are true at that state. 

We sometimes refer to an action profile $\jact = (\act_{1}, \dots, \act_{n}) \in \AcProf$ as a \emph{decision}, and denote by $\act_i$ the action taken by player $i$.
We also consider \emph{partial} decisions.
For a set of players $C \subseteq \Ag$ and action profile $\jact$, we let $\jact_{C}$ and $\jact_{-C}$ be two tuples of actions, respectively, one for all players in $C$ and one for all players in $\Ag \setminus C$.
We also write $\jact_{i}$ for $\jact_{\{i\}}$ and $ \jact_{-i} $ for $ \jact_{\Ag \setminus \{i\}} $.
For two decisions $\jact$ and $\jact'$, we write $(\jact_{C}, \jact_{-C}')$ to denote the decision where the actions for players in $ C $ are taken from $\jact$ and the actions for players in $ \Ag \setminus C $ are taken from $\jact'$.

A \emph{path} $\pi = (s_0, \jact^0), (s_1, \jact^1), \ldots$ is an infinite sequence in $(\St \times \AcProf)^{\omega}$ such that $\trnFun(s_k, \jact^k) = s_{k + 1}$ for all $k$.
%
Paths are generated in the arena by each player~$i$ selecting a {\em
	strategy} $\strElm_i$ that will define how to make choices over
time.  We model strategies as finite state machines with output.
Formally, for arena $A$, a strategy
$\strElm_{i} = (Q_{i}, q_{i}^{0}, \delta_i, \tau_i) $ for player $i$
is a finite state machine with output (a transducer), where $Q_{i}$ is
a finite and non-empty set of \emph{internal states}, $ q_{i}^{0} $ is
the \emph{initial state},
$\delta_i: Q_{i} \times \AcProf \rightarrow Q_{i} $ is a deterministic
\emph{internal transition function}, and
$\tau_i: Q_{i} \rightarrow \Ac_i$ an \emph{action function}, $ \Ac_i \subseteq \Ac $ for all $ i \in \Ag $. Let
$\StrSet_i$ be the set of strategies for player $i$.
%
%
A \emph{strategy profile} $\strpElm = (\strElm_1, \dots, \strElm_n)$ is a vector of strategies, one for each player.
As with actions, $\strpElm_{i}$ denotes the strategy assigned to player $i$ in profile $\strpElm$.
Moreover, by $(\strpElm_{B}, \strpElm'_{C})$ we denote the combination
of profiles where players in disjoint $B$ and $C$ are assigned their corresponding strategies in $\strpElm$ and $\strpElm'$, respectively.

Once a state $s$ and a strategy profile $\strpElm$ are fixed, the game has an \emph{outcome}, a path in $A$, which we denote by $\pi(\strpElm, s)$. 
Because strategies are deterministic, $\pi(\strpElm, s)$ is the unique path induced by $\strpElm$, that is, the sequence $(s_0, \jact^0), (s_1, \jact^1), \ldots$ such that 
\begin{itemize}
	\item $s_{k + 1} = \trnFun (s_k, \jact_{k})$, and 
	\item	$\jact_{k + 1} = (\tau_1(q^k_1), \ldots, \tau_n(q^k_n))$, for all $k \geq 0$. 
\end{itemize}
	Where $q^{k + 1}_i = \delta_i(q^k_i, (\tau_1(q^k_1), \ldots, \tau_n(q^k_n)))$ is the unique sequence of internal states of strategy $\strElm_{i}$ in $\strpElm$ obtained by feeding the result of previous computation at each step.

Arenas define the dynamic structure of games (the actions that agents can perform and their consequences), but lack the feature of games that gives them their strategic nature: players' preferences.
A \emph{multi-player game} is obtained from an arena $A$ by
associating each player with a \emph{goal}. As indicated above, previous work has considered players with goals expressed as \LTL\ formulae, with the idea being that an agent will act as best they can to ensure their \LTL\ goal is satisfied (taking into account the fact that other players will act likewise). In the present article, we consider both goals that are expressed as $\GRone$ formulae, and mean payoff ($\MP$) goals: 
\begin{itemize}
    \item A multi-player \GRone game is a tuple $\Game_{\GRone} = \tuple{A, (\gamma_i)_{i \in \Ag}}$ where $A$ is an
arena and $ \gamma_i $ is the \GRone goal for player $i$.  
\item A multi-player \MP game is a tuple
$\Game_{\MP} = \tuple{A, (\wFun_{i})_{i \in \Ag}}$, where $A$ is an
arena and $\wFun_{i}: \St \to \SetZ$ is a function mapping every state
of the arena into an integer.
\end{itemize}
When it is clear from the context, we refer to a multi-player \GRone or \MP game as a {\em game} and denote it by $\Game$.
In any game with arena $A$, a path $\pi$ in $A$ induces a sequence $\lambda(\pi) = \lambda(s_0) \lambda(s_1) \cdots$ of sets of atomic propositions; if, in addition, $A$ is the arena of an \MP game, then, for each player~$i$, the sequence $\wFun_i(\pi) = \wFun_i(s_0) \wFun_i(s_1) \cdots$ of weights is also induced. 

For a \GRone game and a path $\pi$ in it, the payoff of a player~$i$ is $\pay_i(\pi) = 1$ if $\lambda(\pi) \models \gamma_i$ and $\pay_i(\pi) = 0$ otherwise.
Regarding an \MP game, the payoff of player~$i$ is $\pay_i(\pi) = \MP(\wFun_{i}(\pi))$.
Moreover, for a \GRone game and a path $\pi$, by $\WinSet(\pi) = \{i \in \Ag: \lambda(\pi) \models \gamma_i \}$ and $\LoseSet(\pi) = \{j \in \Ag: \lambda(\pi) \not\models \gamma_j \}$ we denote the set of \emph{winners} and \emph{losers}, respectively, over $\pi$, that is, the set of players that get their goal satisfied and not satisfied, respectively, over $\pi$.
With an abuse of notation, we sometime denote $\WinSet(\strpElm, s) = \WinSet(\pi(\strpElm, s))$ and $\LoseSet(\strpElm, s) = \LoseSet(\pi(\strpElm, s))$, respectively, the set of winners and losers over the path generated by strategy profile $\strpElm$ when starting the game from $s$. Furthermore, we simply write $ \pi(\strpElm) $ for $ \pi(\strpElm,s_0) $.

\vspace*{4pt} \noindent \textbf{Nash equilibrium.}
Using payoff functions, we can define the concept of Nash equilibrium~\cite{OR94}. 
For a game $\Game$, a strategy profile
$\strpElm$ is a \emph{Nash equilibrium} of~$\Game$ if, for every player~$i$ and strategy $\strElm'_i \in \StrSet_i$, we have
$$
\pay_i(\pi(\strpElm))	\geq	\pay_i(\pi((\strpElm_{-i},\strElm'_i))) \ . 
$$
Let $\NE(\Game)$ be the set of Nash equilibria of~$\Game$.

\vspace*{4pt} \noindent \textbf{\enash and rational verification.}
In rational verification, a key question/problem is \enash, which is concerned with the existence of a Nash equilibrium that fulfils a given temporal specification $\varphi$.
Formally, \enash is defined as follows:

\begin{definition}[\enash]
	Given a game $\Game$ and a formula $\varphi$:
	\begin{center}
		Does there exist $\strpElm \in \NE(\Game)$ such that $\pi(\strpElm)
		\models \varphi$?
	\end{center}
\end{definition}

Previous work~\cite{GutierrezHW15,GutierrezHW17-aij,GHW15,GutierrezHW17-apal} has demonstrated that, if we assume player goals are expressed as \LTL\ formulae, the \enash\ problem is 2\exptimeC, and hence highly intractable. 
Motivated by this, in this article, we study
\enash for a number of relevant instantiations of the problem, which
we show to have better (lower) computational complexity.  In particular, we
study cases where
\begin{itemize}
	\item 
	Specifications $\varphi$ are \LTL and players' goals are \GRone;
	\item 
	Specifications $\varphi$ are \LTL and players have \MP goals;
	\item 
	Both the specification $\varphi$ and the goals are \GRone;
	\item 
	Specifications $\varphi$ are \GRone and players have \MP goals. 
\end{itemize}

\noindent \textbf{Automata.}  Some of the algorithms we present for
the \enash problem use techniques from automata theory. Specifically,
we use deterministic automata on infinite words with \textit{Streett}
acceptance conditions. Formally, a \emph{deterministic Streett
	automaton on infinite words} (DSW) is a tuple
$\Automaton = (\Sigma, Q, q^0, \delta, \Omega)$ where $\Sigma$ is the
input alphabet, $Q$ is a finite set of states,
$ \delta: Q \times \Sigma \rightarrow Q $ is a transition function,
$ q^0 $ is an initial state, and $\Omega$ is a Streett acceptance
condition. A Streett condition $ \Omega $ is a set of pairs
$ \{(E_1,C_1),\dots,(E_n,C_n)\} $ where $ E_k \subseteq Q$ and
$ C_k \subseteq Q $ for all $ k \in [1,n] $. A run $ \rho $ is
accepting in a DSW~$\Automaton$ with condition $ \Omega $ if $\rho$
either visits $ E_k $ finitely many times or visits $ C_k $ infinitely
often, {\em i.e.}, if for every $ k $ either
$ \mathit{inf}(\rho) \cap E_k = \varnothing $ or
$ \mathit{inf}(\rho) \cap C_k \neq \varnothing $.


\section{Games of General Reactivity of Rank 1}
\label{sec:GRonegames}

We consider two variations of \GRone games: in the first, the specification formula is expressed in \LTL, while the goals are in \GRone; in the second, both the specification formula and the goals belong to \GRone.  We begin by providing a general result characterizing Nash Equilibrium for \GRone, which is given in terms of \emph{punishments}.  We first require some notation.

For a \GRone game $\Game$, player $j \in \Ag$, and state $s \in \St$, the strategy profile $\vec{\strElm}_{-j}$ is \emph{punishing} for player $j$ in $s$ if $\pi((\vec{\strElm}_{-j}, \strElm_{j}'), s) \not\models \gamma_j$, for every possible strategy $\strElm_{j}'$ of player $j$.
We say that a state $s$ is punishing for $j$ if there exists a punishing strategy profile for $j$ on $s$.
Moreover, we denote by $\Pun_{j}(\Game)$ the set of punishing states in $\Game$.
A pair $(s, \jact) \in \St \times \AcProf$ is \emph{punishing-secure} for player $j$, if $\trnFun(s, (\jact_{-j},\act'_j)) \in \Pun_{j}(\Game)$ for every action $\act_j'$.

\begin{theorem}
	In a given \GRone game $\Game$, there exists a Nash Equilibrium if and only if there exists an ultimately periodic path $\pi$ such that, for every $k \in \SetN$, the pair $(s_k, \jact^k)$ of the $k$-th iteration of $\pi$ is punishing-secure for every $j \in \LoseSet(\pi)$.
\end{theorem}


\begin{proof}[Proof sketch]
	The proof proceeds by double implication.
	
	From left to right, let $\vec{\sigma} \in \NE(\Game)$ and $\pi$ be the ultimately periodic path generated by $\vec{\sigma}$.
	Assume by contradiction that $ \pi $ is not punishing-secure for some $j \in \LoseSet(\pi)$, that is, there is $ k \in \SetN $ and action $ \act'_j $ such that $ \trnFun(s_k,(\jact_{-j},\act'_j)^k) \notin \Pun_{j}(\Game) $.
	Thus, $ j $ can deviate at $ s_k $ and satisfy $ \gamma_j $, which is a contradiction to $ \vec{\sigma} $ being a Nash equilibrium.

	From right to left, recall that $ \pi $ can be generated by a finite transducer, say $\Trans^{\pi} = \tuple{T, t_0, \delta^{\pi}, \tau^{\pi}}$ with $\delta^{\pi}: T \times \AcProf \to T$ being the internal function and $\tau^{\pi}: T \to \AcProf$ being the action function that generates $\pi$.
	Moreover, observe that such transducer can be decomposed into strategies $\strElm_{i}^{\pi} = \tuple{T, t_0, \delta^{\pi}, \tau_{i}^{\pi}}$ where $\tau_{i}^{\pi}(t) = \tau^{\pi}(t)_{i}$.
	Moreover, for every losing player $j \in \LoseSet(\pi)$, there is a memoryless punishing strategy profile $\strElm_{-j}^{\pun}: \St \to \AcProf_{-j}$ for $j$ in every $ s \in \Pun_{j}(\Game)$.
	Such strategy can also be decomposed and \emph{distributed} to the agents different from $j$ as $\strElm_{-j}^{\pun, i}(s) = \strElm_{-j}^{\pun}(s)_{i}$ for every $i \in \Ag \setminus \{j\}$.
	
	Now, for every agent $i$, consider the strategy $\strElm_{i} = \tuple{Q_i, q^0_i, \delta_i, \tau_{i}}$ defined as follows:
	
	\begin{itemize}
		\item
		$Q_i = T \times S \times (\{\top\} \cup \LoseSet(\pi))$;
		
		\item
		$q_i^0 = (t_0, s_0, \top)$;
		
		\item
		$\delta_i$ is defined as follows~\footnote{Note that we should define the internal and action functions on their entire domains.
			However, their definition for the other cases is irrelevant in the proof.}:
		
		$
		\delta_i(t, s, \top, \jact) =
		\begin{cases}
		(\delta^{\pi}(t, \jact), \trnFun(s, \jact), \top), & \text{ if } \jact = \tau^{\pi}(t) \\
		(\delta^{\pi}(t, \jact), \trnFun(s, \jact), j), & \text{ if } \jact_{-j} = (\tau^{\pi}(t))_{-j} \text{ and } \jact_{j} \neq (\tau^{\pi}(t))_{j}
		\end{cases}
		$

		$
		\delta_i(t, s, j, \jact) = (\delta^{\pi}(t, \jact), \trnFun(s, \jact), j)
		$
		
		\item
		
		$\tau_{i}(t, s, \iota) = 
		\begin{cases}
		\tau^{\pi}_{i}(t) & \text{ if } \iota = \top \\
		\strElm_{ - \iota}^{\pun, i}(s) & \text{ otherwise }
		\end{cases}$
	\end{itemize}
	
	Intuitively, the strategy $\strElm_{i}$ mimics the transducer $\Trans^{\pi}$ to produce the play $\pi$.
	In addition to this, it keeps track of the actions taken by the losing agents, checking whether they adhere to the transducer or they deviate unilaterally from it.
	In case of a deviation of agent $j$, the strategy $\strElm_{i}$ flags the deviating agent and switches from mimicking $\Trans^{\pi}$ to adopting the punishment strategy $\strElm_{j}^{\pun}$.
	
	We need to show that the strategy profile $\strpElm$ is a Nash Equilibrium.
	
	Clearly, as $\pi(\strpElm) = \pi$, all the agents that are winning over $\pi$ do not have a beneficial deviation.
	For a losing agent $j$, observe that a unilateral deviation $\strElm_{j}$ triggers the strategy profile $\strpElm_{-j}$ to implement a punishment over $j$.
	Moreover, observe that \GRone objectives are prefix-independent, which implies that the punishment takes effect no matter at which instant of the computation is started being adopted.
	Therefore, every deviation $\strpElm_{j}'$ cannot be beneficial for agent $j$, and hence $\strpElm$ is a Nash Equilibrium.
	\qed
\end{proof}

With this result in place, the following procedure can be seen to solve \enash:

\begin{enumerate}
	\item
	Guess a set $W \subseteq \Ag$ of winners;
	
	\item 
	For each player~$j \in L = \Ag \setminus W$, a loser in the game, compute its punishment region $\Pun_{j}(\Game)$;
	
	
	\item 
	Remove from $\Game$ the states that are not punishing for players $j \in L$ and the edges $(s,s')$ that are labelled with an action profile $\jact$ such that $(s,\jact)$ is not punishing-secure for some $j \in L$, thus obtaining a game $\Game^{-L}$;
	
	\item 
	Check whether there exists an ultimately periodic path $\pi$ in $\Game^{-L}$ such that $\pi \models \varphi \wedge \bigwedge_{i \in W} \gamma_i$ holds. 
\end{enumerate}
Expressed more formally, the above procedure yields Algorithm~1.
\begin{algorithm}[h]
	\caption{\label{alg:enashgrone} \enash of \GRone games.}
	\textbf{Input}: A game $\Game_{\GRone}$ and a specification formula $\varphi$.
	
	
	\For{$i \in \Ag$}{
		Compute $\Pun_i(\Game)$	
	}
	\For{$W \subseteq \Ag$}{
		Compute $L = \Ag \setminus W$
		
		Compute $\Game^{-L}$
		
		\If{$\pi \models (\varphi \wedge \bigwedge_{i \in W} \gamma_i)$ for some $\pi \in \Game^{-L}$}{
			\Return $\mthfun{Accept}$
		}
	}
	\Return $\mthfun{Reject}$
\end{algorithm}

While line~6 requires solving the model checking problem for an \LTL
formula, which can be done in polynomial space, line~5 can be done in
polynomial time. Line~4, on the other hand, makes the procedure run in
exponential time in the number of players, but still in polynomial
space. We then only need to consider line~3: this step can be done in
polynomial time, as we now show.
\begin{theorem}\label{thm:ptimewregion}
	For a given \GRone game $\Game$ over the arena $A = \tuple{\Ag, \Ac, \St, s_0, \trnFun, \labFun}$ and a player $j \in \Ag$, computing the punishing region $\Pun_{j}(\Game)$ of player~$ j $ can be done in polynomial time with respect to the size of both $\Game$ and $\gamma_j$.
\end{theorem}

\begin{proof}
	We reduce the problem to computing the winning region of a suitably defined Streett game with a single pair as the winning condition, whose complexity is known to be $O(mn^{k + 1}k k!)$~\cite{PP06}.
	Given that in our case we have $k = 1$, we obtain a polynomial time algorithm.
	
	Recall that the goal of player $j$ is of the form:
	
	$$
	\gamma_j = \bigwedge_{l = 1}^{m_j}\always \sometime \psi_{l}^{j} \to \bigwedge_{r = 1}^{n_j} \always \sometime \theta_{r}^{j}\text{,}
	$$
	where $\psi_{l}^{j}$'s and $\theta_{r}^{j}$'s are boolean combinations of atomic propositions.
	Then, consider the arena $A' = \tuple{\Ag, \Ac, \St', s_0', \trnFun'}$~\footnote{We omit the definition of labelling function, as not needed here.} where
	
	\begin{itemize}
		\item
		$\St' = \St \times \{0, \ldots, m_j\} \times \{0, \ldots, n_j \}$;
		
		\item 
		$s_0' = (s_0, 0, 0)$;
		
		\item 
		$\trnFun'((s, \iota_{1}, \iota_{2}), \jact) = (\trnFun(s,\jact), \iota_{1}', \iota_{2}')$ where
		
		$\iota_{1}' =
		\begin{cases}
		(\iota_{1} \oplus_{(m_{j} + 1)} 1), & \text{if } \iota_{1} = 0 \text{ or } s \models \psi_{\iota_{1}}^{j}. \\
		\iota_{1}, & \text{otherwise}.
		\end{cases}$\\
		$\iota_{2}' =
		\begin{cases}
		(\iota_{2} \oplus_{(n_{j} + 1)} 1), & \text{if } \iota_{2} = 0 \text{ or } s \models \theta_{\iota_{2}}^{j}.\\
		\iota_{2}, & \text{otherwise. }
		\end{cases}$\\
		And by $\oplus_{k}$ we denote the addition modulo $k$.
	\end{itemize}

	Intuitively, arena $A'$ mimics the behaviour of $A$ and carries two indexes, $\iota_{1}$ and $\iota_{2}$.
	Index $\iota_{1}$ is increased by one every time the path visits a state that satisfies $\psi_{\iota_{1}}^{j}$ and resets to $0$ every time the path visits a state that satisfies $\psi_{m_j}^j$.
	Clearly, $\iota_{1}$ is reset infinitely many times if and only if the path satisfies every $\psi_{l}^{j}$ infinitely many times, and so if and only if it satisfies the temporal specification $\bigwedge_{l = 1}^{m_j}\always \sometime \psi_{l}^{j}$.
	The same argument applies to index $\iota_{2}$, but with respect to the boolean combinations $\theta_{r}^j$'s.
	
	Now, consider the sets $C_j = \St \times \{0\} \times \{0, \ldots, n_j \}$ and $E_j = \St \times \{0, \ldots, m_j \} \times \{0\}$.
	Clearly, the Streett pair $(C_j, E_j)$ is satisfied by all and only the paths in $A'$ that satisfy $\gamma_j$.
	Therefore, the winning region of $\gamma_j$ can be computed as the winning set of the Streett game with $(C_j, E_j)$ being the only Streett pair.
	Observe that the winning region is computable as Street games are \emph{determined}.
	Moreover, having a number of pairs fixed, the computation can be done in polynomial time, which proves our statement. \qed
\end{proof}

Based on Theorem~\ref{thm:ptimewregion}, we have the following result. 

\begin{corollary}
	\label{cor:enashgroneltl}
	The \enash problem for \GRone games with an \LTL specification is \pspaceC.
\end{corollary}

\begin{proof}
	The upper-bound follows from the procedure described above.
	Regarding the lower-bound, note that model-checking an \LTL formula $\varphi$ against a Kripke structure $\Kripke$ can be easily encoded as an instance of \enash where $\Game$ is played over a Kripke structure $\Kripke$, taken to be its arena, players' goals being tautologies, and the specification being $\neg \varphi$. In such a case, we have that $\Kripke \models \varphi$ if and only if \enash for the pair $(\Game,\varphi)$ has a negative answer. \qed 
\end{proof}

Corollary~\ref{cor:enashgroneltl} sharply contrasts with the complexity of \enash\ when goals expressed as \LTL formulae: in this more general case, \enash is 2\exptimeC.

\vspace*{1ex}\noindent \textbf{The special case of \GRone specifications.}
One of hardest parts of Algorithm~\ref{alg:enashgrone} is line~6, where an \LTL model checking problem must be solved, thereby making the running time of the overall procedure exponential in the size of the specification and goals of the players. As we show in the reminder of this section, one way to drastically reduce the complexity of our decision procedure is to require that the specification is also expressed in \GRone.
In such a case, the \LTL model checking procedure in line~6 of Algorithm~\ref{alg:enashgrone} can be avoided, leading to a much simpler construction, which runs in polynomial time for every fixed number of players. In this section, we provide precisely such a simpler construction. 

Recall that every \GRone specification $\varphi$ can be regarded as a Streett condition with a single pair over an arena $A'$ suitably constructed from the original arena~$A$~\cite{BCGHJ10}.
Thus, by denoting $(C_\varphi, E_\varphi)$ and $(C_{i}, E_{i})$ the Streett pairs corresponding to the \GRone conditions $\varphi$ and $\gamma_i$, respectively, the problem of finding a path in $A'$ satisfying the formula $\varphi \wedge \bigwedge_{i \in W} \gamma_i$ amounts to deciding the emptiness of the Streett automaton $\Automaton = \tuple{\AcProf, \St', s_{0}', \trnFun, \Omega}$ where $\Omega = \{(C_\varphi, E_\varphi), (C_{\gamma_i}, E_{\gamma_i})_{i \in W} \}$.

Note that the size of $A'$ is polynomial in the size of the \GRone formulae involved, polynomial in the number of states and actions in the original arena $A$, and exponential in the number of players.
More specifically, we have that $\card{\St'} = \card{\St} \cdot \card{\gamma}^{\card{\Ag}}$ and so the number of edges is at most $\card{\St'}^{2}$.
Moreover, the emptiness problem of a deterministic Streett word automaton can be solved in time that is polynomial in the automaton's index and its number of states and transitions~\cite{HT96,Kupferman15}.
The complexity of the \enash problem takes $2^{\card{\Ag}}$ times a procedure for computing at most $\card{N}$ punishing regions (that is polynomial in the size of both $\Game$ and $\varphi, \gamma_{1}, \ldots, \gamma_{N}$) plus the complexity of the emptiness problem for a Streett automaton whose size is polynomial in $\Game$ $\varphi, \gamma_{1}, \ldots, \gamma_{N}$, and exponential in the number of players.

Based on the constructions described above, we have the following (fixed-parameter tractable) complexity result. 

\begin{theorem}
	\label{thm:enashgronegrone}
	For a given \GRone game $\Game$ and a \GRone formula $\varphi$, the \enash problem can be solved in time that is polynomial in $\card{\St}$, $\card{\Ac}$, and $\card{\varphi}$, $\card{\gamma_{1}}, \ldots, \card{\gamma_{N}}$ and exponential in the number of players $\card{\Ag}$.
	Therefore, the problem is fixed-parameter tractable, parametrized in the number of players.
\end{theorem}



\section{Mean-Payoff Games}
\label{sec:mpgames}
We now focus on multi-player mean-payoff (\MP) games.
As in the previous case, we first characterise the Nash Equilibria of a game in terms of punishments and then reduce \enash to a suitable path-finding problem in the underlying arena.
To do this, we first need to recall the notion of secure values for mean-payoff games~\cite{UW11}.

For a player $i$ and a state $s \in \St$, by $\pun_i(s)$ we denote the
punishment value of $i$ over $s$, that is, the maximum payoff that $i$
can achieve from $s$, when all other players behave adversarially.
Such a value can be computed by considering the corresponding
two-player zero-sum mean-payoff game~\cite{ZP96}.  Thus, it is in
${\np} \cap {\mathsf{co}\np}$, and note that both player $i$ and coalition
$\Ag \setminus \{i\}$ can achieve the optimal value of the game using
memoryless strategies.

For a player $i$ and a value $z \in \SetR$, a pair $(s, \jact)$ is $z$-secure for $i$ if $\pun_i(\trnFun(s, (\jact_{-i}, \act'_i))) \leq z$ for every $\act'_i \in \Ac$.

\begin{theorem}
	\label{thm:pthfinding}
	For every \MP game $\Game$ and ultimately periodic path $\pi = (s_0, \jact_{0}), (s_1, \jact_{1}), \ldots $, the following are equivalent
	
	\begin{enumerate}
		\item 
		There is $\strpElm \in \NE(\Game)$ such that $\pi = \pi(\strpElm, s_0)$;
		
		\item 
		There exists $\vec{z} \in \SetR^{\Ag}$, where $z_{i} \in \{\pun_i(s): s \in \St \}$ such that, for every $i \in \Ag$
		
		\begin{enumerate}
			\item 
			for all $k \in \SetN$, the pair $(s_k, \jact^k)$ is $z_i$-secure for $i$, and 
			
			\item 
			$z_i \leq \pay_i(\pi)$.
		\end{enumerate}
	\end{enumerate}
	
\end{theorem}

%

\begin{proof}
	The proof proceeds by double implication.
	
	For the case $(1) \Rightarrow (2)$, assume that $\strpElm \in \NE(\Game)$ is such that $\pi(\strpElm) = \pi$.
	Thus, define $z_i = \max \set{\pun_i(\trnFun(s_k,(\jact^{k}_{-i}, \act')))}{k \in \SetN, \act' \in \Ac_{i}}$, that is, the max value agent $i$ can achieve by unilaterally deviating from any point in $\pi$ and getting immediately punished.
	By definition, we obtain that $(s_k, \jact^{k})$ is $z_i$-secure for $i$, at every $k \in \SetN$.
	Moreover, assume by contradiction that $\pay_{i}(\pi) < z_i$ for some agent $i$.
	Then, let $k \in \SetN$ and $\act'_i \in \Ac_{i}$ be such that $z_i = \pun_i(s_k, (\jact_{-i}, \act'_{i}))$.
	Thus, there exists a strategy $\strElm_{i}'$ that follows $\strpElm_{i}$ for $k$ steps and then deviates using $\act_{i}'$ that ensures a payoff of $z_i$ for agent $i$.
	Such strategy is a beneficial deviation of agent $i$ from $\strpElm$, in contradiction with the fact that $\strpElm$ is a Nash Equilibrium.
	
	For the case $(2) \Rightarrow (1)$, we define a strategy profile $\strpElm$ and then prove it is a Nash Equilibrium.
	First observe that, being $\pi$ ultimately periodic, there exists a finite transducer $\Trans^{\pi} = \tuple{T, t_0, \delta^{\pi}, \tau^{\pi}}$ with $\delta^{\pi}: T \times \AcProf \to T$ being the internal function and $\tau^{\pi}: T \to \AcProf$ being the action function that generates $\pi$.
	Moreover, observe that such transducer can be decomposed into strategies $\strElm_{i}^{\pi} = \tuple{T, t_0, \delta^{\pi}, \tau_{i}^{\pi}}$ where $\tau_{i}^{\pi}(t) = \tau^{\pi}(t)_{i}$.
	In addition to this, for every agent $j$, consider the \emph{memoryless} strategy $\strElm_{-j}^{\pun}: \St \to \AcProf_{-j}$ that minimizes the payoff of agent $j$ in every state $s \in \St$.
	Such strategy can also be decomposed and \emph{distributed} to the agents different from $j$ as $\strElm_{-j}^{\pun, i}(s) = \strElm_{-j}^{\pun}(s)_{i}$ for every $i \in \Ag \setminus \{j\}$.
	Now, for every agent $i$, consider the strategy $\strElm_{i} = \tuple{Q_i, q^0_i, \delta_i, \tau_{i}}$ defined as follows:
	
	\begin{itemize}
		\item
			$Q_i = T \times S \times (\{\top\} \cup \Ag \setminus \{i\})$;
			
		\item
			$q_i^0 = (t_0, s_0, \top)$;
			
		\item
			$\delta_i$ is defined as follows:
			
				$
				\delta_i(t, s, \top, \jact) =
				\begin{cases}
				(\delta^{\pi}(t, \jact), \trnFun(s, \jact), \top), & \text{ if } \jact = \tau^{\pi}(t) \\
				(\delta^{\pi}(t, \jact), \trnFun(s, \jact), j), & \text{ if } \jact_{-j} = (\tau^{\pi}(t))_{-j} \text{ and } \jact_{j} \neq (\tau^{\pi}(t))_{j}
				\end{cases}
				$

				$\delta_i(t, s, j, \jact) = (\delta^{\pi}(t, \jact), \trnFun(s, \jact), j)$
			
		\item
			$\tau_{i}(t, s, \iota) = 
			\begin{cases}
			\tau^{\pi}_{i}(t) & \text{ if } \iota = \top \\
			\strElm_{ - \iota}^{\pun, i}(s) & \text{ otherwise }
			\end{cases}$~\footnote{Note that we should define the internal and action functions on their entire domains.
			However, their definition for the other cases is irrelevant in the proof.}
		
	\end{itemize}

	Intuitively, the strategy $\strElm_{i}$ mimics the transducer $\Trans^{\pi}$ to produce the play $\pi$.
	In addition to this, it keeps track of the actions taken by the other agents, checking whether they adhere to the transducer or they deviate unilaterally from it.
	In case of a deviation of agent $j$, the strategy $\strElm_{i}$ flags the deviating agent and switches from mimicking $\Trans^{\pi}$ to adopting the punishment strategy $\strElm_{j}^{\pun}$.
	Clearly, the strategy profile $\strpElm = \tuple{\strElm{1}, \ldots, \strElm_{n}}$ is such that $\pi(\strpElm) = \pi$.
	It remains to show that it is a Nash Equilibrium.
	Note that, since every strategy $\strElm_{i}$ adopts the punishment for agent $j$ at every possible deviation.
	Note that, being $\mp$ a prefix independent condition, the payoff for agent $j$ is punished no matter at which instant the punishment strategy is started being adopted.
	At this point, being every pair $(s_k, \jact^k)$ in $\pi$ $z_j$-secure for agent $j$, it holds that every deviation of agent $j$ does not ensure a payoff greater than $z_j$, that is $\pay_{j}(\strpElm_{-j}, \strpElm_{j}') \leq z_j$.
	On the other hand, from condition (b) of item 2 in the statement, we have that $z_j \leq \pay_{j}(\strpElm)$.
	By putting these two conditions together, we obtain
	
	\[
	\pay_{j}(\strpElm_{-j}, \strpElm_{j}') \leq z_j \leq \pay_{j}(\strpElm)\text{.}
	\]
	This proves that every deviation of agent $j$ from $\strpElm$ is not beneficial, and so that $\strpElm$ is a Nash Equilibrium. \qed
\end{proof}

The characterization of Nash Equilibria provided in Theorem~\ref{thm:pthfinding} allows us to turn the \enash problem for \MP games into a path finding problem over $\Game$. 
Similarly to the case of \GRone games, we have the following procedure.

\begin{enumerate}
	\item 
	For every $i\in\Ag$ and $s \in \St$, compute the value $\pun_i(s)$;
	
	\item 
	Guess a vector $z \in \SetR^{\Ag}$ of values, each of them being a punishment value for a player $i$;
	
	\item 
	Compute the game $\Game{[z]}$ by removing the states $s$ such that $\pun_i(s) \leq z_i$ for some player $i$ and the transitions $(s, \jact)$ that are not $z_i$ secure for some player $i$;
	
	\item
	Find an ultimately periodic path $\pi$ in game $\Game{[z]}$ such that $\pi \models \varphi$ and $z_i \leq \pay_i(\pi)$ for every player $i \in \Ag$.
\end{enumerate}

Step~1 can be done in \np for every pair $(i,s)$, step~2 can be done in exponential time and polynomial space in the number of $z$-secure values, and step~3 can be done in polynomial time, similar to the case of  \GRone games.
Regarding the last step, its complexity depends on the specification language.
For the case of $\varphi$ being an \LTL formula, consider the formula 
$$\varphi_{\enash} := \varphi \wedge \bigwedge_{i \in \Ag} (\MP(i) \geq z_i) \text{,}$$
written in the language $\LTL^{\mthfun{Lim}}$, an extension of \LTL where statements about mean-payoff values over a given weighted arena can be made~\cite{BCHK14}.
Observe that formula $\varphi_{\enash}$ corresponds exactly to requirement $2(b)$ in Theorem~\ref{thm:pthfinding}.
Moreover, since every path in $\Game{[z]}$ satisfies condition $2(a)$ by construction, every path that satisfies $\varphi_{\enash}$ is a solution of the \enash problem and \emph{vice versa}.
We can solve the latter problem by model checking the formula against the arena underlying $\Game{[z]}$.
Since this can be done in \pspace~\cite{BCHK14}, we have the following result. 

\begin{corollary}
	\label{cor:enashmpltl}
	The \enash problem for \MP games with an \LTL  specification formula $\varphi$ is \pspaceC.
\end{corollary}


As for the case of \GRone games, we can summarize the procedure in the following algorithm (Algorithm~\ref{alg:enashmp}).

\begin{algorithm}
	\caption{\label{alg:enashmp} \enash of \MP games.}
	\textbf{Input}: A game $\Game_{\MP}$ and a specification formula $\varphi$.
	
	
	\For{$i \in \Ag$ and $s \in \St$}{
		Compute $\pun_i(\Game)$	
	}
	\For{$\vec{z} \in \{\pun_i(s): s \in \St\}^{\Ag}$}{
		Compute $\Game{[z]}$
		
		\If{$\pi \models \varphi_{\enash}$ for some $\pi \in \Game{[z]}$}{
			\Return $\mthfun{Accept}$
		}
	}
	\Return $\mthfun{Reject}$
\end{algorithm}

\vspace*{1ex}\noindent \textbf{The special case of \GRone specifications.}
As in the case of \GRone games, here we show that restricting the specification language to \GRone also lowers the complexity for \MP games. The reason for this is that the path finding problem for \GRone specifications can be done while avoiding model-checking an $\LTL^{\mthfun{Lim}}$ formula. 
In order to do this, we follow a different approach. Using an \MP game $\Game$ and a \GRone specification $\phi$ we define a linear program such that the linear program has a solution if and only if the pair $(\Game,\phi)$ is an instance of \enash. In particular, this approach is similar to the technique used in \cite[Theorem 2]{GMPRW17}, where Linear Programming is used to find the complexity of solving a variant of \enash. 
Formally, we have the following result. 

\begin{theorem}
	\label{thm:enashmpgrone}
	The \enash problem for \MP games with a \GRone specification $\varphi$ is \npC.
\end{theorem}

\begin{proof}
	We will define a linear program of size polynomial in $\Game$ having a solution if and only if there exists an ultimately periodic path whose payoff for every player $i$ is at least a minimum threshold $z_i$ and satisfies the \GRone specification.
	
	In order to do that, first recall that $\varphi$ has the following form
	$$
	\varphi = \bigwedge_{l = 1}^{m} \always \sometime \psi_{l} \to \bigwedge_{r = 1}^{n} \always \sometime \theta_{r}\text{,}
	$$
	and let $V(\psi_{l})$ and $V(\theta_r)$ be the subset of states in $\Game$ that satisfy the Boolean combinations $\psi_{l}$ and $\theta_{r}$, respectively.
	Observe that property $\varphi$ is satisfied over a path $\pi$ if, and only if, either $\pi$ visits every $V(\theta_r)$ infinitely many times or visits some of the $V(\psi_{l})$ only a finite number of times.
	
	For the game $\Game{[z]}$, let $\tuple{V, E, (\wFun_{i}')_{i \in \Ag}}$ be the underlying graph, where $ \wFun_{i}'(v) = \wFun_{i}(v) - z_i $ for every $ i \in \Ag$, and $v \in V \subseteq \St$.
	Furthermore, for every edge $e\in E$, we introduce a variable $x_e$.
	Informally, the value $x_e$ is the number of times that the edge $e$ is used on a cycle.
	Formally, let:
	
	\begin{itemize}
		\item
			$\src(e) = \{v \in V : \exists w\, e = (v,w) \in E\}$;
			
		\item
			$\trg(e) = \{v \in V : \exists w\, e = (w,v) \in E\}$;
			
		\item
			$\OUT(v) = \{e \in E : \src(e) = v\}$;
		
		\item
			$\IN(v) = \{e \in E : \trg(e) = v\}$.
	\end{itemize}
	
	Consider $\psi_{l}$ for some $1 \leq l \leq m$, and define the linear program $\LP(\psi_{l})$ with the following inequalities and equations:
	\begin{enumerate}
		\itemsep1.3em
		\item[Eq1:]
		$x_e \geq 0$ for each edge $e$ 
		
		\hfill a basic consistency criterion;
		
		\item[Eq2:]
		$\Sigma_{e \in E} x_e \geq 1$ 
		
		\hfill ensures that at least one edge is chosen;
		
		\item[Eq3:]
		for each $i \in \Ag$, $\Sigma_{e \in E} \wFun_i'(\src(e)) x_e \geq 0$ 
		
		\hfill ensures that the total sum of any solution is positive;
		
		\item[Eq4:]
		$\Sigma_{\src(e) \cap V(\psi_{l}) \neq \emptyset} x_e = 0$ 
		
		\hfill ensures that no state in $V(\psi_{l})$ is in the cycle associated with the solution;
		
		\item[Eq5:]
		for each $v \in V$, $\Sigma_{e \in \OUT(v)} x_e = \Sigma_{e \in \IN(v)} x_e$  
		
		\hfill says that the number of times one enters a vertex is equal to the number of times one leaves that vertex.	
	\end{enumerate}
	
	By construction, it follows that $\LP(\psi_{l})$ admits a solution if and only if there exists a path $\pi$ in $\Game$ such that $z_i \leq \pay_i(\pi)$ for every player $i$ and visits $V(\psi_{l})$ only {\em finitely many times}.
	Note that the condition $z_i \leq \pay_i(\pi)$ is ensured by Eq3.
	Indeed, the value of a path $\pi$ in $\Game{[z]}$ that is represented in a solution to $\LP(\psi_{l})$, and thus satisfying Eq3, is such that $0 \leq \pay^{\Game{[z]}}_{i}(\pi)$, with $pay^{\Game{[z]}}_{i}$ representing the payoff function for agent $i$ in the game $\Game{[z]}$.
	Now observe that, as the weights in $\Game{[z]}$ are all downshifted by a value $z_i$ for every agent $i$, it holds that $\pay_{i}(\pi) = \pay^{\Game{[z]}}_{i}(\pi) + z_i$, which in turns implies that $z_i \leq \pay_i(\pi)$.
	
	Now, consider also the linear program $\LP(\theta_{1}, \ldots, \theta_{n})$ defined with the following inequalities and equations: 
	
	\begin{enumerate}
		\itemsep1.3em
		\item[Eq1:]
		$x_e \geq 0$ for each edge $e$ 
		
		\hfill a basic consistency criterion;
		
		\item[Eq2:]
		$\Sigma_{e \in E} x_e \geq 1$ 
		
		\hfill 
		ensures that at least one edge is chosen;
		
		\item[Eq3:]
		for each $i \in \Ag$, $\Sigma_{e \in E} \wFun_i'(\src(e)) x_e \geq 0$ 
		
		\hfill ensures that the total sum of any solution is positive;
		
		\item[Eq4:]
		for all $1 \leq r \leq n$, $\Sigma_{\src(e) \cap V(\theta_{r}) \neq \emptyset} x_e \geq 1$ 
		
		\hfill ensures that for every $V(\theta_{r})$ at least one state is in the cycle;
		
		\item[Eq5:]
		for each $v \in V$, $\Sigma_{e \in \OUT(v)} x_e = \Sigma_{e \in \IN(v)} x_e$  
		
		\hfill says that the number of times one enters a vertex is equal to the number of times one leaves that vertex.	
	\end{enumerate}
	
	In this case, $\LP(\theta_{1}, \ldots, \theta_{n})$  admits a solution if and only if there exists a path $\pi$ such that $z_i \leq \pay_i(\pi)$ for every player $i$ and visits every $V(\theta_{r})$ {\em infinitely many times}.
	
	Since the constructions above are polynomial in the size of both $\Game$ and $\phi$, we can conclude it is possible to check in \np the statement that there is a path $\pi$ satisfying $\varphi$ such that $z_i \leq \pay_i(\pi)$ for every player~$i$ in the game if and only if one of the two linear programs defined above has a solution. For the lower bound, we use \cite{UW11} and observe that if $\phi$ is true, then the problem is equivalent to checking whether the \MP game has a Nash equilibrium. \qed
\end{proof}



\section{Social welfare verification}
\label{secn:social-welfare}

Until this point, the problems considered primarily concerned about
the satisfaction of a temporal logic property $\varphi$ over the game
$\Game$.  However, one might be interested in achieving an outcome
that is somehow best also for the agent society.  To capture this
setting, we introduce \emph{social welfare} measures. Social welfare
measures are \emph{aggregate} measures of utility. Thus, a social welfare
measure takes as input a profile of utilities, one for each player in
the game, and somehow aggregates these into an overall measure, indicating how
good the outcome is for society as a whole. Note that since social welfare is
inherently a quantitative measure, in this section we restrict our attention to \MP games.

Formally, for a game $\Game$ with a set $N$ of agents, a \emph{social
  welfare function} $\sw$ takes the form
\begin{center}
	$\sw: \SetR^{N} \to \SetR$
\end{center}
\noindent Thus, a social welfare function maps a $N$-tuple of real numbers into a real number which represents the \emph{aggregated payoff}.  More specifically, for a
strategy profile $\strpElm$, the social welfare of $\strpElm$ is given
by $\sw(\pay_{1}(\strpElm), \ldots, \pay_{N}(\strpElm))$.  With an
abuse of notation, we denote $\sw(\strpElm)$ the social welfare of
$\strpElm$.  Many different social welfare functions have been
proposed in the literature of economic theory.  Here, we confine out
attention to the two best known: \emph{utilitarian} and
\emph{egalitarian} social welfare. These functions are defined as
follows:

\begin{itemize}
	\item
		The \emph{utilitarian} social welfare function is given by $\usw(\strpElm) = \sum_{i\in N} \pay_{i}(\strpElm)$.
		
	\item
		The \emph{egalitarian} social welfare function is given by $\esw(\strpElm) = \min_{i\in N} \{\pay_{i}(\strpElm)\}$.
\end{itemize}

For simplicity, for a given game $\Game$ and a formula $\varphi$, by $\enash_{\Game}(\varphi) = \set{\strpElm \in \NE}{\pi(\strpElm) \models \varphi}$ we denote the set of Nash equilibria that satisfy $\varphi$, that is, that are a solution to the \enash problem of $(\Game, \varphi)$.
For a fixed social welfare function $\sw$ on a game $\Game$, by:

\begin{itemize}
	\item
		$\MaxENash_{\sw}(\Game,\varphi) = \max_{\strpElm \in \enash_{\Game}(\varphi)}\{\sw(\strpElm)\}$, and
		
	\item
		$\MinENash_{\sw}(\Game,\varphi) = \min_{\strpElm \in \enash_{\Game}(\varphi)}\{\sw(\strpElm)\}$
\end{itemize}
\noindent we denote the maximal and minimal social welfare achieved over a Nash equilibrium profile, respectively, satisfying a given specification $\varphi$.

The values of $\MaxENash$ and $\MinENash$ determine how good or bad
the \enash solutions are from the perspective of the agents in the
game collectively.  Here, we consider both the decision and function
problem.

\begin{definition}[Threshold social welfare]
	For a given \MP game $\Game_{\MP}$, a social welfare function $\sw$, and a threshold value $t$, decide whether there exists a strategy profile $\strpElm$ in $\enash_{\Game}(\varphi)$ such that $t \leq \sw(\strpElm)$.
	In case of a positive answer to this decision question, the pair $(\Game, \varphi)$ is called $t$\emph{-increase}.
	
	Analogously, decide whether there exists a strategy profile $\strpElm$ in $\enash_{\Game}(\varphi)$ such that $t \geq \sw(\strpElm)$.
	In case of a positive answer to this decision question, the pair $(\Game, \varphi)$ is called $t$\emph{-decrease}.
\end{definition}

\begin{definition}[Max and Min social welfare]
	For a given \MP game $\Game_{\MP}$ and a social welfare function $\sw$, compute $\MaxENash_{\sw}(\Game, \varphi)$ and $\MinENash_{\sw}(\Game, \varphi)$.
\end{definition}

The two definitions above can be instantiated with many different social welfare functions. In the following two subsections, we consider them in the context of the utilitarian and egalitarian welfare measures defined above.


\subsection{Social welfare computation with \LTL specifications}

We first show how to check that a given \MP game $\Game_{\MP}$ and a \LTL specification meets a given threshold $t$.
As the utilitarian and egalitarian functions require different proofs, we address them separately.
For the utilitarian function, we have the following.

\begin{theorem}\label{thm:ltl-sw}
	For a given \MP game $\Game_{\MP} = \tuple{A, (\wFun_{i})_{i \in \Ag}}$, an \LTL specification $\varphi$, and a threshold value $t$, deciding whether there exists a strategy profile $\strpElm \in \enash_{\Game}(\varphi)$ such that $t \leq \usw(\strpElm)$ is \pspaceC.
	Analogously, deciding whether there exists a strategy profile $\strpElm \in \enash_{\Game}(\varphi)$ such that $t \geq \usw(\strpElm)$ is \pspaceC.
\end{theorem}

\begin{proof}
	It is enough to show the case $t \leq \usw(\strpElm)$ as the other one is similar.
	The solution is a slight modification of the \enash problem for \MP games with \LTL specifications.
	Consider the arena $A' = \tuple{\Ag \cup \{n + 1\},  \Ac, \St, s_0, \trnFun', \labFun}$ with $\trnFun'$ defined as
	
	\[
	\trnFun'(\act_{1}, \dots, \act_{n}, \act_{n + 1}) = \trnFun(\act_{1}, \dots, \act_{n})
	\]
	\noindent for every $(\act_{1}, \dots, \act_{n}, \act_{n + 1}) \in \Ac^{\card{\Ag} + 1}$, and the \MP game $\Game_{\MP}' = \tuple{A', (\wFun_{i})_{i \in \Ag}, (\wFun_{n+1})}$ with $\wFun_{n+1}(s) = \sum_{i\in \Ag}(\wFun_{i}(s))$ for every $s\in \St$.
	
	Intuitively, we have included an extra agent in the game, having no effect/impact on the executions, in a way that it carries information about the social welfare of the original game.
	Indeed, observe that, for every strategy profile $\strpElm$ in $\Game_{\MP}'$, it holds that 
	
	\[
	\pay_{n + 1}'(\strpElm) = \sum_{i\in \Ag}\pay_{i}'(\strpElm) = \sum_{i\in \Ag}\pay_{i}(\strpElm_{-(n + 1)}) = \usw(\strpElm_{-(n + 1)})
	\]

	We can employ the same construction for solving the \enash problem for \MP games with \LTL specifications to solve the threshold problem.
	It suffices to replace the \LTLlim formula $\varphi_{\enash}$ with 
	
	\[
	\varphi_{\enash}^{\usw, t} := \varphi_{\enash} \wedge \MP(n + 1) \geq t\text{.}
	\]
	
	The computational complexity of the procedure is \pspace as for \enash.
	The lower bound easily follows from the model checking of \LTL. \qed
\end{proof}
For the case of egalitarian social welfare, we have the following.

\begin{theorem}
	For a given \MP game $\Game_{\MP} = \tuple{A, (\wFun_{i})_{i \in \Ag}}$, an \LTL specification $\varphi$, and a threshold value $t$, deciding whether there exists a strategy profile $\strpElm \in \enash_{\Game}(\varphi)$ such that $t \leq \esw(\strpElm)$ is \pspaceC.
	Analogously, deciding whether there exists a strategy profile $\strpElm \in \enash_{\Game}(\varphi)$ such that $t \geq \esw(\strpElm)$ is \pspaceC.
\end{theorem}

\begin{proof}
	It is enough to show the case $t \leq \esw(\strpElm)$ as the other one is similar.
	As for the case of utilitarian social welfare functions, the solution is a slight modification of the \enash problem for \MP games with \LTL specifications.
	Indeed, observe that we can specify that the payoff of agent $i$ is greater than the threshold $t$ by the \LTLlim formula $\MP(i) \geq t$.
	Therefore, specifying that the egalitarian social welfare is at least $t$ can be done by the conjunction $\bigwedge_{i \in \Ag} \MP(i) \geq t$.
	Thus, it suffice to replace the \LTLlim $\varphi_{\enash}$ for the \enash problem with
	
	\[
	\varphi_{\enash}^{\esw, t} := \varphi_{\enash} \wedge \bigwedge_{i \in \Ag} \MP(i) \geq t \text{.}
	\]
	
	Again, the computational complexity of the procedure is \pspace and the lower bound follows from the model checking of \LTL.
	\qed
\end{proof}

\subsection{Social welfare computation with \GRone specifications}
%
In this section, we address social welfare threshold problems with \GRone specifications. The techniques are similar to the ones used in the case of \LTL specifications. Firstly, we consider the utilitarian social welfare function. For a given \MP game $\Game_{\MP} = \tuple{A, (\wFun_{i})_{i \in \Ag}}$, we build the arena $A'$ and the game $\Game_{\MP}'$ analogous to the way it is done in the proof of Theorem~\ref{thm:ltl-sw}. Now, to solve the case $t \leq \usw(\strpElm)$, we adapt the procedure for solving \enash for \MP games with \GRone specifications (Theorem~\ref{thm:enashmpgrone}) as follows. We construct the corresponding multi-weighted graph $ W = \tuple{V,E,(\wFun_i')_{i \in \Ag \cup {n+1} }}$ where $ \wFun_{n+1}'(v) = \wFun_{n+1}(s) - t $. Then, solving \enash problem for such an instance corresponds exactly to the threshold social welfare problem $t \leq \usw(\strpElm)$. For the case $t \geq \usw(\strpElm)$, we simply define $ \wFun_{n+1}'(v) = t - \wFun_{n+1}(s) $. To obtain the lower bounds, we reduce from the \enash problem for \MP games with \GRone specifications. For the case $ t \leq \usw(\strpElm) $, we set $ t = \min \{ \wFun_{n+1}(s) : s \in \St \} $, and the other case, we fix $ t = \max \{ \wFun_{n+1}(s) : s \in \St \} $. Thus, we obtain the following result.

\begin{theorem}
	For a given \MP game $\Game_{\MP} = \tuple{A, (\wFun_{i})_{i \in \Ag}}$, a \GRone specification $\varphi$, and a threshold value $t$, deciding whether there exists a strategy profile $\strpElm \in \enash_{\Game}(\varphi)$ such that $t \leq \usw(\strpElm)$ is \npC.
	Analogously, deciding whether there exists a strategy profile $\strpElm \in \enash_{\Game}(\varphi)$ such that $t \geq \usw(\strpElm)$ is \npC.
\end{theorem}

Now we turn our attention to the egalitarian social welfare function. To solve the social threshold problem $ t \leq \esw(\strpElm) $, we directly adapt from the procedure for solving \enash for \MP games with \GRone specifications (Theorem~\ref{thm:enashmpgrone}). For the game $\Game{[z]}$, we build the underlying graph $\tuple{V, E, (\wFun_{i}')_{i \in \Ag}}$ where $ \wFun_{i}'(v) = \wFun_{i}(s) - (\max\{z_i,t\}) $. Then we define the linear programs $\LP(\psi_l)$ and $\LP(\theta_1,\dots,\theta_n)$ in the same way. Observe that, one of the two linear programs has a solution if and only if there is a path $\pi$ satisfying $\varphi$ such that for every player $i$, $z_i \leq \pay_i(\pi) $ and $t \leq \pay_i(\pi) $. To obtain the lower bound, again, we reduce from the \enash problem for \MP games with \GRone specifications. The reduction simply follows from the fact that by fixing $t = \min \{ \wFun_{i}(s) : i \in \Ag, s \in \St \}$, we can encode \enash problem into the social threshold problem. The case $ t \geq \esw(\strpElm)$ is similar. Therefore, we obtain the following result.

\begin{theorem}
	For a given \MP game $\Game_{\MP} = \tuple{A, (\wFun_{i})_{i \in \Ag}}$, a \GRone specification $\varphi$, and a threshold value $t$, deciding whether there exists a strategy profile $\strpElm \in \enash_{\Game}(\varphi)$ such that $t \leq \esw(\strpElm)$ is \npC.
	Analogously, deciding whether there exists a strategy profile $\strpElm \in \enash_{\Game}(\varphi)$ such that $t \geq \esw(\strpElm)$ is \npC.
\end{theorem}

The threshold social welfare calculation can be used to approximate the \MaxENash and \MinENash values of a game, be it either utilitarian or egalitarian.
Note that, for every agent $i \in \Ag$ and every strategy profile $\strpElm$ in the game, it holds that 

\[
\min(\wFun_i) = \min_{s \in \St}\{\wFun_{i}(s)\} \leq
\pay_{i}(\strpElm) \leq
\max_{s \in \St}\{\wFun_{i}(s)\} = \max(\wFun_{i})\text{.}
\]
\noindent This establishes a bound also on the social welfare function, which is given by 

\[
\sum_{i\in \Ag} \min(\wFun_i) \leq
\MinENash_{\sw}(\Game, \varphi) \leq
\MaxENash_{\sw}(\Game, \varphi) \leq
\sum_{i\in \Ag} \max(\wFun_i)\text{.}
\]

Moreover, observe that, for two values $t < t'$, if $(\Game, \varphi)$ is $t$-increase but not $t'$-increase, then it holds that $t \leq \MaxENash_{\sw}(\Game, \varphi) < t'$.
Analogously, if $(\Game, \varphi)$ is $t'$-decrease, but not $t$-decrease, then it holds that $t \leq \MinENash_{\sw}(\Game, \varphi) < t'$.

These observations allow to apply a bisection-like method to approximate $\MaxENash$ and $\MinENash$.
Moreover, note that at each iteration of the method, the absolute error is halved, which ensures linear convergence of the method~\cite{Sik82}.
Particularly, we obtain an approximation of the values within a fixed tolerance $\epsilon > 0$ in a number $n$ of iterations bounded by $n_{\epsilon} = \lceil \log_{2}(\frac{b - a}{\epsilon}) \rceil$, with $a = \sum_{i\in \Ag} \min(\wFun_i)$ and $b = \sum_{i\in \Ag} \max(\wFun_i)$.



\section{Other Rational Verification Problems}\label{secn:more}
\enash is, we believe, the most fundamental problem in the rational
verification framework, but it is not the only one. The two other key
problems are \anash and \rvnonemptiness. The former is the dual
problem of \enash, which asks, given a game $\Game$ and a
specification $\phi$, whether $\phi$ is satisfied in {\em all\/} Nash
equilibria of $\Game$. The latter simply asks whether the game
$\Game$ has at least one Nash equilibrium, and it can be thought of as the special case of \enash
where the specification $\phi$ is any tautology.

We can conclude from (the proofs of) the results presented so far, which are summarised in Table~\ref{tab:results}, that while \anash for \GRone games is also \pspace and \fpt, respectively, in case of \LTL and \GRone specifications, for \MP games the problem is, respectively, \pspace and co\np, in each case.
In addition, we can also conclude that whereas \rvnonemptiness for \GRone games is \fpt, for \MP games is \npC.
These results contrast with those when players' goals are general \LTL formulae, where all problems are 2\exptimeC since \LTL synthesis, which is 2\exptime-hard~\cite{PnueliR89}, can be encoded.
These results also contrast with those presented in~\cite{GaoGW17}, where it is shown that, in succinct model representations given by iterated Boolean games or reactive modules, all problems in the rational verification framework can be polynomially reduced to \rvnonemptiness, which clearly cannot be the case here, unless the whole polynomial hierarchy collapses. 



\section{Concluding Remarks}
\label{sec:conc}

We have presented improved complexity results for rational verification problems in three different settings: in the analysis of response properties of reactive systems modelled as multiagent systems; verification of mean-payoff games; and verification of collective properties of multiagent systems through the analysis of social welfare properties. The first scenario mostly concerns the verification of qualitative properties of reactive systems; the second the verification of quantitative properties; and the third the verification of ``community'' properties, as opposed to individual properties of agents in a system. In the remainder of this article, we discuss further the impact and relevance of our results in these three areas. 

\paragraph{\bf Reactive systems} The logical analysis of reactive systems is typically carried out using either linear temporal logics, such as \LTL, or branching time temporal logics, such as \CTL\ and \CTLstar. Such analysis may involve verifying that a temporal logic property holds in a given system (model checking) or automatically constructing the system from a temporal logic specification (automated synthesis). Rational verification subsumes both problems, and applies to systems modelled in a distributed way as a collection of semi-autonomous agents (a multiagent system). Despite the greater scope of rational verification with respect to both model checking and automated synthesis, previous work has shown that the overall complexity of rational verification is typically not higher/worse than the combined complexity of the associated synthesis problem. This connection also transfers when considering goals expressed in the \GRone fragment of LTL, where an initial solution in 2\exptime is reduced to complexities lying in the polynomial hierarchy. However, to do so, careful attention must be paid to how the additional game-theoretic analysis that rational verification entails must be done without blowing up the combined computational complexity. This is particularly important since, in rational verification, strategies for multiple agents must be synthesised, rather than a single model for a reactive system.

\paragraph{\bf Mean-Payoff games} In the computer science literature, mean-payoff games have been considered as a way of understanding the long-term behaviour (the average performance) of a system---the most common setting is that of a two-player game in which one of the players model the system and the other player models the environment. From a game-theoretic point of view, these are two-player games, which in a perfect information setting can be solved in NP${}\cup{}$coNP, thus without a known polynomial time algorithm to solve them. In case of rational verification with mean-payoff objectives, the problem is definitely harder, (unless P=NP, which is unlikely). We have shown that if the principal has an LTL goal, the problem matches the complexity of \LTL\ model checking, a complexity gap that cannot be avoided since \LTL\ model checking is a particular case. But, even with \GRone specifications, the problem is very likely to be strictly harder than solving (two-player perfect-information) mean-payoff games since we have shown that with mean-payoff objectives the problem is \np-Complete.

\paragraph{\bf Social Welfare} While rational verification tends to privilege the preferences of individual agents in a system, social welfare measures focus, instead, on what is considered to be best for a society of agents. Because of this, our results regarding social welfare outcomes may complement nicely the analysis performed in rational verification as originally defined, where the perfromance of society as a whole was irrelevant. We have shown that even in this scenario, better complexity results can be achieved with respect to the complexity of the problem when only individual preferences are considered, as in a Nash equilibrium. In the specific scenario that we considered in the paper, we have shown that the problem is \pspace-complete, and therefore still efficient with respect to the space complexity of the problem.

\paragraph{\bf Future Work} A limitation in adopting widely the use of rational verification instead of other reasoning techniques is its combined complexity, which is closely related to the complexity of associated automated synthesis problems.
Our results are important because they show that for several significant settings, rational verification can be done with polynomial space algorithms.
These results are much more attractive than in the general case, and hold out the hope of efficient practical tools (\emph{c.f.}\ the \emph{Equilibrium Verification Environment} (EVE)~\cite{GutierrezNPW18,GutierrezNPW20}, a tool for the automated analysis of temporal equilibrium properties).
Further practical implementations thus seem to be a natural step forward towards the deployment of rational verification in more realistic scenarios. 



\section*{Acknowledgements}
Wooldridge gratefully acknowledges the support of the ERC under Advanced Grant 291528 (``RACE''), and the support of the Alan Turing Institute in London.  
Najib acknowledges the support of ERC Starting Grant 759969 (AV-SMP).
Perelli acknowledges the support of the ERC project ``WhiteMech'' (grant agreement No 834228) and the EU ICT-48 2020 project TAILOR (No. 952215).

%
%

\bibliographystyle{spmpsci}      
\bibliography{references}   


\end{document}